\documentclass[11pt]{article}

\usepackage{amsmath, amssymb, amsthm}
\usepackage[margin=1in]{geometry}
\usepackage{graphicx}
\usepackage[round]{natbib}
\usepackage[hidelinks]{hyperref}

\usepackage{setspace}
\newtheorem{definition}{Definition}

\newtheorem{example}{Example}
\newtheorem{assumption}{Assumption}
\newtheorem{proposition}{Proposition}

\newtheorem{theorem}{Theorem}

\newcommand{\E}{\mathbb{E}}

\title{Learning Where to Look: The Geometry of Data Combination and Collection}
\author{Giacomo Opocher\thanks{\scriptsize University of Mannheim, giacomo.opocher@uni-mannheim.de. \textbf{This version is a work in progress; feel free to contact me if you have any comments.} This paper greatly benefited from the guidance of Silvia Sarpietro and Davide Viviano, and meaningful discussions with Yagan Hazard, Toru Kitagawa, Tobias Großboltin, Christoph Rothe, Liyang Sun, and Juan C. Yamin. I also thank all seminar participants at the University of Bologna for their insightful comments. All mistakes are my own.}}\date{\today}

\begin{document}
\maketitle

%\begin{abstract}
%    A researcher observes a population in the status quo and wants to learn from a small, possibly selected sample fully exposed to an intervention, its effect on the rest of the population.
%    How should they combine these sources of information? 
%    And who should receive the intervention next to improve on estimation the most?
%    I focus on convex matching estimators and answer both questions by adopting a statistical decision theoretic perspective and drawing new connections with results from computational geometry.
%    First, I show that interpolation over the Delaunay simplex containing a target point minimizes worst-case bias.
%    Second, I provide robust confidence intervals and show that inference-optimal matching rules lie on a frontier connecting the proposed estimator to a minimum-norm Synthetic Control estimator.
%    Finally, I characterize the minimax refinement problem and show that the Delaunay triangulation of the existing data provides sufficient information to solve it numerically and restricts the space \textit{where to look} for optimal additions.
%    In a semi-synthetic application in development economics, I show that Delaunay matching achieves a lower worst-case bias compared to commonly used estimators in practice (e.g. $\sim7$ times smaller than synthetic controls) and that the minimax collection plan decreases it by $\sim 90\%$ within 20 additions.
%\end{abstract}

\begin{abstract}
    Large datasets covering populations in the status quo are becoming increasingly available, while interventions can usually be rolled out only to small selected samples. 
    How can these two sources of information be combined to estimate the effect on untreated units? 
    And who should be treated next to improve estimation the most?
    I focus on matching estimators and shed light on both problems drawing from results in computational geometry.
    First, I prove that interpolation over the Delaunay simplex minimizes worst-case bias at any untreated unit, given any treated sample.
    Second, I show the collection problem has a geometric representation that yields significant reductions in worst-case bias with few additions ($\sim 80\%$ with 20 additions in an empirical application), leading to sizeable cost-efficiency gains in collection plans.
\end{abstract}
\clearpage

\section{Introduction}

Large administrative datasets that cover substantial shares of populations in a status quo are becoming increasingly available to researchers.
Examples span several fields, including political economy \citep[e.g.][]{brown_et_al_2023,brown_et_al_2025}, public economics \citep[e.g.][]{bergeron_tourek_weigel_2024}, labor economics \citep[e.g.][]{chetty_et_al_2026_atlas, amorim_et_al_2026}, and development economics \citep[e.g.][]{asher_novosad_2020,britto_pinotti_sampaio_2022}.
At the same time, interventions whose effect is ex ante ambiguous are commonly rolled out to a small, possibly selected portion of the full population \citep[e.g.][]{chetty_et_al_2026_digital_twins}. 
This practice may be motivated by learning the effects before scaling up the policy to prevent large welfare losses, or by tight budget constraints researchers may face.

How can we combine these two sources of information to learn the effects of the intervention of interest? And, given an initial sample of treated units, who should be treated next to improve estimation most?

In this paper, I study these problems from a statistical decision-theoretic perspective \citep{wald_1950}. 
Suppose units in the population are defined by a continuous covariate space.
A researcher can observe the treated potential outcome at $n$ covariate values, and the untreated outcomes at sufficiently many values that their conditional expectations can be considered known. 
For instance, suppose a published large-scale study \citep[e.g.][]{chetty_et_al_2026_atlas} lets the researcher observe the population in the status quo, but their budget only allows them to assign an intervention of their own interest to a small portion of it.
They aim to estimate the effect of the intervention on any untreated unit. 
The researcher commits to a matching rule.
They first estimate a target unit's outcome under the intervention with a weighted average of the treated units' observed outcomes. 
Then, they take the difference with the known expected outcome under the status quo.
Unless the relationship between the covariate values and the outcome space under treatment is linear, the researcher's estimates will be biased.
The magnitude of the bias depends on the function defining this relationship which is unknown to the researcher, who only knows it has bounded second derivatives.

Given any target point and any $n$ covariate points we can observe under treatment, what matching rule minimizes the largest possible bias?
And who should get treated next to reduce it the most? 

I shed light on these questions by drawing new connections with old results from computational geometry \citep{delaunay_1934,rajan_1994,waldron_1998}.
To build intuition around the problem, suppose the covariate space is two-dimensional.
Further, suppose a target point falls into the interior (i.e. the convex hull) of the $n$ points collected (henceforth, the collection design).
Then, for $n > 3$, there are infinitely many ways to match the target covariate value as a weighted average of the collection design, leaving the first decision problem unsolved and the second undefined.

As a first contribution, I prove that only one rule minimizes the worst-case bias at any target point: assigning barycentric weights to the vertices of the \textit{Delaunay triangle} containing it.
Delaunay triangles were first introduced by \citet{delaunay_1934} and are commonly used to approximate terrain elevations in geographical maps.
Three points in the collection design form a Delaunay triangle if and only if the triangle's circumcircle does not contain any other point. 
The set of Delaunay triangles then defines the Delaunay triangulation of the collection design.
Under minimal conditions, the Delaunay triangles form a unique triangulation covering the convex hull of the collection design.
Therefore, any target point belonging to the convex hull of the collection design falls into one of these triangles. 
I call the estimator that assigns barycentric weights to its vertices the Delaunay Matching Estimator (DME).

%I provide robust confidence intervals for the DME estimate and characterize the frontier of inference-optimal weighting rules over classes of causal laws.
%By inference-optimal, I mean a weighting rule that minimizes the length of the robust confidence set given the class of causal laws, that is, given a bound on its second derivative.
%I show that the frontier connects DME, which is inference-optimal as the bound converges to infinity, to a minimum-norm Synthetic Control estimator, which is inference-optimal under linear causal laws.

As a second contribution, I provide a geometric characterization of the collection problem: given the starting treated sample of $n$ units, which $n+1$\textit{th} unit decreases the largest worst-case bias over target points the most?
Define a minimum-enclosing circle as the smallest circle that contains a triangle.
The largest worst-case bias over all target points is proportional to the radius of the largest minimum-enclosing circle in the Delaunay triangulation of the collection design. 
Therefore, the additional design point that decreases the largest worst-case bias the most is the point in the covariate space that decreases such a radius the most in the refined triangulation.
Moreover, any optimal addition must lie inside the circumcircle containing the triangle whose minimum-enclosing circle is currently largest.
Otherwise, by definition of Delaunay triangles, such a triangle does not change with the addition and the largest worst-case bias does not decrease.
The Delaunay triangulation of the collection design thus at the same time (i) determines the improvement of any candidate and (ii) restricts the region \textit{where to look} for an optimal addition. 
I call the resulting collection plan the Geometric Collection Plan.

Finally, I present an empirical application based on the setting in \citet{bergeron_tourek_weigel_2024}.
In the original paper, the authors have access to a large-scale field experiment that randomizes a property tax-rate reduction to properties in the Democratic Republic of the Congo and estimates its effect on tax revenues.
I suppose the authors did not have access to this large-scale experiment, but could only observe the large sample in the status quo ($\sim 5.5$K properties) and rolled out the tax-rate increase to only $80$ properties.
First, I show that DME delivers lower worst-case bias at any target property than minimum-norm, entropy-balancing, and adaptive-nearest-neighbor estimators, whose median worst-case biases are $5.3$, $5.6$, and $1.24$ times higher, respectively.
Then, I show that the Geometric Collection Plan reduces the maximum worst-case bias by $5.5\%$ assigning one additional property to the tax-rate reduction and by $82.4\%$ assigning twenty.
By contrast, random sampling produces no reduction, leading to sizeable gains in data collection efficiency.

\paragraph{Related Literature.}
First, this paper contributes to the literature on matching by proving Delaunay Matching minimaxity for absolute bias, and connecting this set of estimators to minimax collection plans.

Matching methods originally constructed comparisons between treated and untreated units with similar observed characteristics, drawing on the propensity-score framework of \citet{rosenbaum_rubin_1983} and its econometric development by \citet{heckman_ichimura_todd_1997,heckman_ichimura_todd_1998}.
\citet{rosenbaum_1989} formulated the construction of matched samples as an optimization problem, while \citet{abadie_imbens_2006,abadie_imbens_2011} characterized the large-sample properties of nearest-neighbor matching and developed bias corrections.
\citet{stuart_2010} and \citet{imbens_rubin_2015} provide general treatments of this literature.
Synthetic controls extend the comparison to convex combinations of donor observations \citep{abadie_gardeazabal_2003,abadie_diamond_hainmueller_2010}, while balancing-weight methods choose weights subject to covariate-moment restrictions \citep{hainmueller_2012,zubizarreta_2015}.

The closest estimator to the one studied here is the penalized synthetic control of \citet{abadie_lhour_2021}.
They characterize its Delaunay support and show that, as the penalty vanishes, it selects the exactly balanced convex combination with the smallest weighted sum of squared donor distances whenever exact balance is feasible.
Under the same covariate metric and general-position conditions, this limit coincides with Delaunay Matching.
My contribution is to establish an exact finite-sample minimax justification for this choice.
%and show that the frontier of inference-optimal weighting rules coincides in spirit with the penalty rational in their paper.
This work closely relates to \citet{kallus_2020}, who also studies minimax optimal matching. 
The main technical differences are that (i) the target parameter in \citet{kallus_2020} is the sample average treatment effect on the treated (here the average effect on the untreated), (ii) the loss function is conditional mean squared error (here worst-case bias), and (iii) the smoothness condition imposed on the missing potential outcome function is Lipschitz continuity (rather than bounded second derivatives).
In particular, the smoothness assumption considered here allows the connection with \citet{waldron_1998} and \citet{rajan_1994}, opening up new perspectives on this problem while allowing most economically relevant scenarios.
In particular, combining the sharp interpolation bounds of \citet{waldron_1998} with geometric properties of Delaunay triangulations \citep{delaunay_1934,rajan_1994} leads to DME's minimaxity for worst-case bias.
%The minimax optimal weights can be computed directly from the donor triangulation, without solving a numerical worst-case optimization problem at each target or specifying the magnitude of the curvature bound.
%With bounded outcomes, I then combine this sharp bias characterization with Hoeffding's inequality to construct finite-sample confidence intervals.
%Minimizing their length yields a frontier of inference-optimal weighting rules that connects minimum-norm Synthetic Controls, when the causal law is linear, to DME as the curvature bound grows, making explicit the trade-off between worst-case bias and sampling uncertainty.

This paper also contributes to the literature on minimax collection rules and experimental design, providing a geometric characterization of minimax refinements. 
In particular, I prove an exact correspondence between the statistical minimax data-collection problem and a geometric optimization problem that does not require any information other than the collection design itself.

This literature applies \citeauthor{wald_1950}'s \citeyearpar{wald_1950} statistical decision framework to the choice of a sampling process or experimental design.
\citet{dominitz_manski_2017,dominitz_manski_2022} study minimax-regret sample design under trade-offs between collection costs, sample size, and data quality.
\citet{carneiro_lee_wilhelm_2020} choose sample size and covariates under a budget constraint to improve the precision of average treatment-effect estimates.
When the objective is subsequent treatment choice, \citet{manski_tetenov_2016} derive sample-size guarantees for near-optimal treatment rules, and \citet{hu_et_al_2024} characterize minimax-regret allocations of trial participants across heterogeneous subpopulations.
Related work chooses treatment assignments to reduce estimation risk \citep{kasy_2016} or designs experiments to improve subsequent policy choices \citep{kasy_sautmann_2021,higbee_2026}.

Especially close are contributions that ask where to experiment to inform decisions beyond the experimental sample.
\citet{gechter_et_al_2024} study site selection for external validity using a Bayesian decision framework, while \citet{montiel_olea_et_al_2026} establish conditions under which minimax-regret site selection is equivalent to a $k$-median problem.
\citet{breza_chandrasekhar_viviano_2026} combine prediction with abstention to identify contexts where further evidence is valuable, and \citet{epanomeritakis_viviano_2026} jointly choose experiments and sample allocations when experimental estimates complement potentially biased external evidence.
My contribution is to show that the connection with computational geometry can further our understanding of this class of problems.
Combining the optimality of Delaunay Matching with simplex interpolation bounds \citep{waldron_1998} yields an exact geometric representation of the collection problem.
The Delaunay triangulation of the existing collection design provides sufficient information to solve the minimax problem numerically and to significantly restrict the space \textit{where to look} for new additions.

In a broader sense, this paper contributes to the literatures on statistical decision theory and causal inference, highlighting the value of studying the performance of estimators and collection plans from a geometric perspective.

The rest of the paper is organized as follows.
Section~\ref{sec:decision_problem} introduces the statistical decision problem and characterizes worst-case bias for convex, affine-exact matching estimators.
Section~\ref{sec:delaunay_matching} establishes the minimax optimality of Delaunay Matching. 
Section~\ref{sec:collection_design} derives the geometric representation of the one-step collection problem.
Section~\ref{sec:application} illustrates the estimation and collection results in an empirical application in public economics, and Section~\ref{sec:conclusions} concludes.
Proofs are collected in Appendix~\ref{sec:proofs}.

\section{Description of the Statistical Decision Problem}
\label{sec:decision_problem}
Consider a binary intervention that perturbs a status quo.
Denote with $D_i$ the intervention status of an element $i$ of a population of interest.
Denote with $Y_i(0), Y_i(1)$ the potential outcomes under the status quo and the intervention.
Denote with $X_i$ a vector of continuous covariates that takes values $x \in \mathcal{X}\subseteq \mathbb{R}^d$.
Suppose $\mathcal{X}$ is a convex space.
Index with $P$ the joint distribution of $(X_i,Y_i(0),Y_i(1),D_i)$.

\paragraph{Estimand.} A researcher (R) aims to estimate the intervention effect for any target value $x \in \mathcal{X}$ observed in the status quo.
Formally, denote with
\begin{equation}\label{eq:cef_funct}
    \mu_d(x) := \E_P[Y_i(d)|X_i=x, D_i = 0] \quad \text{for}\ d \in \{0,1\}
\end{equation}
the conditional expectation function (henceforth CEF) for units in the status quo.  
Given a target value in a target set $x \in \mathcal{X}_0$, the estimand of interest is:
\begin{equation}
    \tau(x) := \mu_1(x) - \mu_0(x),
\end{equation}
that is, the intervention effect on a unit in the status quo with covariate value $x$.
Denote with $\mathcal{T}$ the values $\tau(x)$ can take.

\begin{assumption}[Potential Outcomes under Intervention] \label{ass:cef_int}
    Suppose that
    \begin{equation}
        Y_i(1) = f(X_i) + U_i, 
    \end{equation}
    where $\mathbb{E}_P[U_i | X_i = x, D_i = d] = 0$, $\forall x \in \mathcal{X}$ and $\forall d \in \{0,1\}$.
    Moreover, suppose
    \begin{equation}
        f \in \mathcal{F}_c:=\left\{f\in C^2(\mathcal X): \sup_{x\in\mathcal{X}}\|\nabla^2 f(x)\|_{\mathrm{op}} \leq c \ \text{for}\ 0<c<\infty\right\}.
    \end{equation}
\end{assumption}
Henceforth, I refer to $f$ as the causal law.
Assumption \ref{ass:cef_int} (i) imposes separability between a function of covariates and a conditional-mean-zero idiosyncratic component; (ii) fixes the causal law across units that did and did not receive the intervention; (iii) imposes a smooth function class for the causal law.
Causal laws with jumps, kinks, or curvature that becomes arbitrarily large on $\mathcal X$ are excluded.
More generally, when $\mathcal X$ is compact, every twice continuously differentiable function satisfies Assumption \ref{ass:cef_int}.
This assumption plays a central role in identifying the worst-case causal law, and computing the worst-case bias. 
Note that, although the bound $c$ needs to be known to compute the worst-case bias, it will not be necessary to find the minimax estimator or the minimax refinement plan.
Denote with $\mathcal{P}_c$ the set of distributions that satisfy Assumption \ref{ass:cef_int}.
The following examples highlight one case where Assumption \ref{ass:cef_int} is satisfied and one where it is violated.

\begin{example}[Quadratic causal law]
Suppose $f(x)=\alpha+\beta'x+x'Qx$, where $Q$ is symmetric and $2\|Q\|_{\mathrm{op}}\leq c$.
If $U_i$ has conditional mean zero and selection into the intervention depends only on $X_i$, then Assumption \ref{ass:cef_int} holds.
In particular, selection may depend on observed covariates.
\end{example}

\begin{example}[Selection on unobserved outcomes]
Suppose $U_i$ can take either sign conditional on $X_i$, but units receive the intervention when $U_i>0$.
Then $\E_P[U_i\mid X_i,D_i=1]>0$, so the same function $f$ does not describe mean potential outcomes for treated and untreated units.
Assumption \ref{ass:cef_int} therefore fails.
\end{example}

\paragraph{Sampling.} A researcher (R) has at their disposal a large sample from the population in the status quo and a finite sample of the same population under the intervention.
\begin{assumption}[Sampling Scheme]\label{ass:samp_status_quo}
The R observes the full population under the status quo, or a sample large enough so that $\mu_0(x)$ can be treated as known $\forall x \in \mathcal{X}$.
Moreover, the researcher observes
\begin{equation}
    S^n \sim_{\mathrm{i.i.d.}} P^n | D_i = 1.
\end{equation}
Denote with $\mathcal{S}^n$ the sample space. 
Denote with $\mathbf{X}^n \in \mathcal{X}_1^n \subset \mathcal{X}^n$ the collection design, that is, the set of values $x$ that are observed under the intervention.
Denote the conditional sampling distribution $P^n|D_i = 1$ as $P^n_1$.
\end{assumption}

Assumption \ref{ass:samp_status_quo} imposes that the conditional expectation function of the potential outcome under the status quo is known, and that a finite sample of treated units is available.
The following example highlights one practical case where the Assumptions \ref{ass:cef_int} and \ref{ass:samp_status_quo} are satisfied.

\begin{example}[Collection following a previous study]\label{ex:coll_previous}
Suppose a large previous study identifies $\mu_0(x)$ for the population of interest. 
The researcher then recruits $n$ units from the same population, treats all of them, and records the same outcome and covariates. 
Recruitment may be random or restricted to a set of covariate values. 
In the latter case, selection may depend on $X_i$ but is independent of $U_i$ conditional on $X_i$.
This setting satisfies Assumption \ref{ass:samp_status_quo}, provided the underlying causal law satisfies Assumption \ref{ass:cef_int}.
\end{example}

\begin{example}[Large survey data and a small Prolific study]
Suppose a large survey records covariates and outcomes for the population under the status quo, allowing $\mu_0(x)$ to be estimated precisely.
For example, the Cooperative Election Study surveys $50{,}000+$ U.S. adults, recording demographic characteristics, political attitudes, turnout, and vote choice \citep{ansolabehere_rivers_2013}.
A smaller random sample is recruited via a survey company and exposed to an intervention.
The survey company cannot cover all covariate values in the population, but ensures that sampling is random within the covered values.
For example, some voters are recruited on Prolific and given the opportunity to use an A.I. informational tool to measure its effect on voting behavior \citep[e.g.][]{ash_galletta_opocher_2025}.
This setting satisfies Assumption \ref{ass:samp_status_quo}, provided the underlying causal law satisfies Assumption \ref{ass:cef_int}.
\end{example}

\paragraph{Decision Problem.} Before observing $\mu_0(x)$ and $S^n$, the researcher specifies a set of estimators $\mathcal{M}$ and commits to a choice
\begin{equation}
    m : \mathcal{S}^n \times \mathcal{U}_0 \times \mathcal{X} \to \mathcal{T}.
\end{equation}
Let 
\begin{equation}
    \hat{\tau}_m(x) := m(S^n, \mu_0, x).
\end{equation}
\begin{definition}\label{def:general_problem}
    The R's decision problem is 
    \begin{equation}
        \min_{m \in \mathcal{M}} \sup_{P \in \mathcal{P}_c} |\E_{P_1^n}[\tau(x) - \hat{\tau}_m(x) | \mathbf{X}^n]|, \quad \forall x \in \mathcal{X}_0 = \mathrm{conv}(\mathbf{X}^n).
    \end{equation}
    That is, finding the estimator that minimizes the worst-case absolute bias conditional on the collection design for all target values in its convex hull.
\end{definition}

This statistical decision problem considers a researcher who aims to find the estimator that minimizes the worst-case bias for \textit{any} target unit and conditional on the collection design.
This criterion informs the researcher about the performance of the estimator for any specific value they may aim to approximate and for the specific collection design they collected.
Therefore, it is a stronger guarantee than the average performance over the target points and collection designs.

To highlight the insights that taking a geometric perspective can bring to this problem, I focus on convex matching estimators whose properties make this link especially transparent.

\begin{definition}[Convex Matching Estimators]
    Define the set of weighting functions
    \begin{equation}
        \Delta := \left\{ \delta : \mathcal{X}_1^n \times \mathcal{X}_0 \to [0,1]^n \text{ such that }\ \sum_{i=1}^n \delta(\mathbf{X}^n,x)_i=1,\ \sum_{i=1}^n \delta(\mathbf{X}^n,x)_iX_i=x. \right\}
    \end{equation}
    as the set of functions that take as input a collection design and a target value and give back a weight to each donor value, such that the weights sum up to one and are linearly precise.
    Furthermore, define
    \begin{equation}
        \mathcal{M}_\delta := \left\{ m_\delta(S^n, \mu_0, x) = \sum_{i = 1}^n \delta(\mathbf{X}^n, x)_i Y_i - \mu_0(x): \delta \in \Delta \right\}
    \end{equation}
    as the set of convex matching estimators.
\end{definition}

The class $\mathcal M_\delta$ is defined as the class of matching estimators that recover the target point exactly on the matching space and only use weights in the unit interval that sum up to 1.
Examples of estimators that belong to this class include synthetic control weights, entropy weighting, and nearest-neighbor estimators.
This class of estimators allows us to simplify the general problem in Definition \ref{def:general_problem} as described in the next theorem.

\begin{proposition}[Minimax Matching] \label{prop:matching_problem}
Suppose Assumptions \ref{ass:cef_int} and \ref{ass:samp_status_quo} hold. 
Then, 
\begin{equation}
    \min_{m \in \mathcal{M}_\delta} \sup_{P \in \mathcal{P}_c} \left|\E_{P_1^n}[\tau(x) - \hat{\tau}_m(x) | \mathbf{X}^n]\right| = \min_{\delta \in \Delta} \frac{c}{2}\sum_{i=1}^n \delta(\mathbf{X}^n, x)_i \cdot ||X_i - x||^2,
\end{equation}
for all $x \in \mathcal{X}_0$.
\end{proposition}

Proposition \ref{prop:matching_problem} shows that, under the restrictions on the potential outcome functions in Assumption \ref{ass:cef_int} and the sampling scheme described in Assumption \ref{ass:samp_status_quo}, and considering the set of convex matching estimators, the problem in Definition \ref{def:general_problem} is equivalent to minimizing over $\Delta$ a distance metric in the matching space scaled by the second derivative bound.

This is a central result for this paper as it allows us to connect the minimax problem to well-known results in computational geometry and prove the minimaxity of Delaunay Matching.

\section{Delaunay Matching Minimaxity}
\label{sec:delaunay_matching}

This section introduces Delaunay Matching and shows that it minimizes worst-case absolute conditional bias at every target point in the collection design's convex hull among convex, affine-exact matching estimators.
Let's first define a Delaunay triangulation.

\begin{definition}[Delaunay triangulation]\label{def:delaunay_triangulation}
A triangulation $T$ of the collection design $\mathbf{X}^n$ is a Delaunay triangulation if every simplex $\sigma\in T$ has a circumsphere whose interior contains no other $X_i$. 
\end{definition}

It is convenient to think of Definition \ref{def:delaunay_triangulation} in a two-dimensional covariate space.
Figure \ref{fig:delaunay_empty_circumcircle} illustrates the concept.
Given the cloud of points that define $\mathbf{X}^n$, three points form a triangle if and only if no other point falls into its circumcircle.

\begin{figure}[h]
  \centering
  \caption{Delaunay Triangulation}
  \label{fig:delaunay_empty_circumcircle}
  \includegraphics[width=0.62\linewidth]{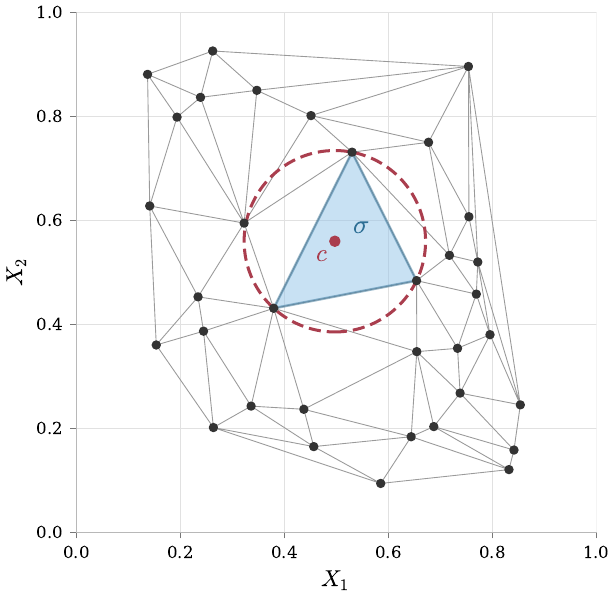}
\smallskip
\begin{minipage}{0.88\linewidth}
  \footnotesize\textit{\textbf{Notes:}} The gray edges depict a Delaunay triangulation of donor covariate locations in $\mathbb R^2$. For the highlighted simplex $\sigma$, the dashed circumcircle has center $c$ and contains no donor point in its interior.
\end{minipage}
\end{figure}

Delaunay triangulations can then be leveraged to define an estimator that, for any point falling in the convex hull of the collection design, assigns barycentric weights to the simplex that contains it.
So, for example, any point contained in $\sigma$ in Figure \ref{fig:delaunay_empty_circumcircle} would be recovered as a weighted average of the three vertices of $\sigma$.
The next definition introduces Delaunay Matching formally.

\begin{definition}[Delaunay Matching Estimator]\label{def:del_matching_est}
    Denote with $\sigma_n(x)$ any simplex of the Delaunay triangulation of $\mathbf{X}^n$ containing $x$, and with $\nu_n(x)$ its set of vertices.
    Delaunay Matching assigns barycentric weights to the vertices of $\sigma_n(x)$ and zero weights to all the other points.
    Formally, define
    \begin{equation}
        \delta_\mathrm{del}(\mathbf{X}^n, x)_i := \begin{cases}
            0 & \quad \text{if} \quad i \notin \nu_n(x) \\
            \lambda(\sigma_n(x), x)_i & \quad \text{if} \quad i \in \nu_n(x),
        \end{cases}
    \end{equation}
    where $\lambda(\cdot)$ is the unique vector that satisfies:
    \begin{equation}
        \sum_{\{i\in \nu_n(x)\}} \lambda(\sigma_n(x), x)_i =1, \quad
        \sum_{\{i\in \nu_n(x)\}} \lambda(\sigma_n(x), x)_iX_i =x, \quad
        \lambda(\sigma_n(x),x)_i \geq 0
    \end{equation}
    Then,
    \begin{equation}
        m_\mathrm{del}(S^n, \mu_0, x) := \sum_{i=1}^n \delta_\mathrm{del}(\mathbf{X}^n, x)_i\cdot Y_i - \mu_0(x).
    \end{equation}
\end{definition}

The following theorem establishes that Delaunay Matching is the unique solution of the decision problem in Definition \ref{def:general_problem}.

\begin{theorem}[Minimax optimality of Delaunay Matching]\label{thm:delaunay_budget}
Under Assumptions \ref{ass:cef_int} and \ref{ass:samp_status_quo}, 
\begin{equation}
    m_\mathrm{del} \in \arg\min_{m \in \mathcal{M}_\delta} \sup_{P \in \mathcal{P}_c} \left|\E_{P_1^n}[\tau(x) - \hat{\tau}_m(x) | \mathbf{X}^n]\right|,
\end{equation}
for every $x\in\mathcal X_0$ and every $\mathbf{X}^n \in \mathcal{X}_1^n$.
\end{theorem}

\paragraph{Proof Sketch.}
I provide insight into Delaunay Matching minimaxity with the following argument visualized in Figure \ref{fig:delaunay_pointwise_argument}. 
Fix a target point $x$. 
Lift each point in the collection design onto a paraboloid (in light blue in Figure \ref{fig:delaunay_pointwise_argument}), with height equal to its squared distance from the origin. 
Any matching rule then determines a point directly above the target, whose vertical distance from the paraboloid (in orange) is proportional to worst-case bias. 
Three lifted points define a plane (light green), and the intersection of this plane with the paraboloid projects onto their circumcircle (in purple). 
If another design point lies inside that circle in the $(x_1, x_2)$ plane, its lift lies below the triangle's plane (in green), allowing a smaller vertical distance.   
Thus, the minimizing triangle must have an empty circumcircle, that is, it must be Delaunay. 
Therefore, Delaunay Matching attains the lowest feasible lifted point above $x$, which proves the result in Theorem \ref{thm:delaunay_budget}. 

\begin{figure}[h]
  \centering
  \caption{Illustration of Delaunay Matching Minimaxity.}
  \label{fig:delaunay_pointwise_argument}
  \includegraphics[width=0.78\linewidth]{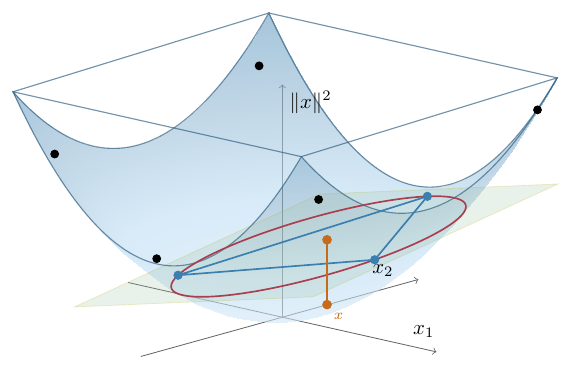}
\smallskip
\begin{minipage}{0.92\linewidth}
  \footnotesize\textit{\textbf{Notes:}} The horizontal axes $(x_1,x_2)$ span the covariate space $\mathcal X\subseteq\mathbb R^2$, while the vertical axis records the value of $\|x\|^2$. 
  Each location $X_i$ is lifted to $(X_i,\|X_i\|^2)$ on the paraboloid. 
  The blue simplex is the lift of the Delaunay simplex $\sigma(x)$ containing the target value $x$, and the green plane is the affine plane that contains any linear combination of its lifted vertices. 
  The crimson curve is the lifted circumcircle, which is the intersection of the paraboloid with that affine plane. 
  Because the Delaunay circumsphere is empty, the plane lies weakly below any convex combination of other lifted donor points. 
  At $x$, its height is $\sum_i\delta_{\mathrm{del}}(\mathbf X^n,x)_i\|X_i\|^2$, whereas the paraboloid has height $\|x\|^2$. 
  Their vertical gap is therefore $\sum_i\delta_{\mathrm{del}}(\mathbf X^n,x)_i\|X_i-x\|^2$. 
  Any other $\delta\in\Delta$ represents a convex combination of lifted donor points with weighted covariate value $x$ and must therefore produce a height weakly above the same plane.
\end{minipage}
\end{figure}

\section{Minimax Refinement}
\label{sec:collection_design}

This section characterizes the one-step minimax collection problem and shows how the Delaunay triangulation determines its loss and restricts where a strictly improving addition can lie.

I consider a researcher who aims to collect one new point in the convex hull of the existing collection design to decrease the highest worst-case loss over possible target points as much as possible.
In other words, the researcher wants to collect the point that decreases the largest mistake they could make the most.
Definition \ref{def:collection_plan} introduces this problem formally.

\begin{definition}[One-step minimax collection plan]
\label{def:collection_plan}
A one-step minimax collection plan is a measurable rule
\begin{equation}
    \pi:\mathcal X_1^n\to\mathcal X_0.
\end{equation}
Given the realized collection design $\mathbf X^n$, define
\begin{equation}
    X_{n+1}:=z,
    \qquad
    \mathbf X^{n+1}:=(X_1,\ldots,X_n,z).
\end{equation}
The one-step minimax collection problem is
\begin{equation}
z^\star
\in
\arg\min_{z\in\mathrm{conv}(\mathbf X^n)}
\sup_{x \in \mathrm{conv}(\mathbf{X}^{n})}
\sup_{P \in \mathcal{P}_c}
\left| \E_{P^{n+1}_1} \left[\tau(x) - \hat{\tau}_{\mathrm{del}}(x) | \mathbf{X}^{n+1}\right] \right|
\end{equation}
The minimax collection plan satisfies
\begin{equation}
    \pi(\mathbf X^n)=z^\star.
\end{equation}
\end{definition}

The following examples provide practical scenarios formalized by Definition \ref{def:collection_plan}.

\begin{example}[Targeted recruitment]
    Building on Example \ref{ex:coll_previous}, suppose that the researcher collects 100 units forming the collection design $\mathbf X^{100}$ by drawing at random from the distribution of $X$.
    After the 100\textit{th} unit, the researcher stops selecting at random and starts recruiting strategically.
    In particular, the researcher chooses the covariate profile $z\in\operatorname{conv}(\mathbf X^{100})$ that minimizes the highest worst-case bias resulting from $\mathbf {X}^{101}$.
    In other words, they start recruiting over $\mathcal{X}_1^{100}$ to reduce the worst-case bias as much as possible.
\end{example}

Next, I introduce a geometric representation of the problem in Definition \ref{def:collection_plan} that (i) allows us to solve the problem using only the collection design and (ii) significantly restricts the space where to look for new additions.
Before introducing the result formally, we need to define minimum-enclosing balls.

\begin{definition}[MEB radius and Circumball]
    For any simplex $\sigma$, denote with $r_{\mathrm{meb}}(\sigma)$ the radius of the Minimum Enclosing Ball (MEB) of its vertices and with $B(\sigma)$ its closed circumball.
\end{definition}
In two dimensions, the MEB is the smallest circle you could put a triangle into. 
This coincides with the circumcircle for acute triangles (see Figure \ref{fig:delaunay_empty_circumcircle}) and with the circle with center on the largest edge and diameter equal to its length for obtuse triangles. 

The next theorem draws from a result in \citet{waldron_1998} to show that the minimax refinement point minimizes the largest MEB's radius in the refined collection design.
This implies that $c$ need not be known to find the optimal point.
Moreover, the theorem states that any optimal refinement needs to belong to the circumcircle of the largest MEB radius simplex of the existing collection design.
In this specific sense, by adopting this geometric perspective, we \textit{learn where to look} for optimal additions.

\begin{theorem}[Geometric representation of the collection loss]
\label{prop:collection_loss}
Given the realized collection design $\mathbf X^n$, let $T(\mathbf X^n)$ be a Delaunay triangulation and choose
\begin{equation}
    \sigma_n^\star\in
    \arg\max_{\sigma\in T(\mathbf X^n)}r_{\mathrm{meb}}(\sigma)^2.
\end{equation}
For any candidate $z\in\operatorname{conv}(\mathbf X^n)$, let $T(\mathbf X^{n+1})$ be a Delaunay triangulation of the augmented design $\mathbf X^{n+1}=(X_1,\ldots,X_n,z)$. Then,
\begin{equation}
\sup_{x\in\operatorname{conv}(\mathbf X^n)}
\sup_{P \in \mathcal{P}_c}
\left| \E_{P^{n+1}_1} \left[\tau(x) - \hat{\tau}_{\mathrm{del}}(x) | \mathbf{X}^{n+1}\right] \right|=
\frac c2
\max_{\sigma\in T(\mathbf X^{n+1})}
r_{\mathrm{meb}}(\sigma)^2.
\end{equation}
Consequently, the exact one-step minimax collection problem in Definition \ref{def:collection_plan} is equivalent to
\begin{equation}
z^\star
\in
\arg\min_{z\in\operatorname{conv}(\mathbf X^n)\cap B(\sigma_n^\star)}
\max_{\sigma\in T(\mathbf X^{n+1})}
r_{\mathrm{meb}}(\sigma)^2.
\end{equation}
\end{theorem}

Figure \ref{fig:waldron_global} illustrates the two parts of Theorem \ref{prop:collection_loss}.
Panel (a) reports the largest worst-case bias after collecting at each candidate $z$, relative to its value before collection.
The yellow regions leave this loss unchanged, while the blue region identifies locations that improve the collection design.
The star marks a minimizing location $z^\star$.
We can see visually the fact that minimizers are contained in the circumcircle of the triangle with the largest MEB radius.
That's because if that triangle doesn't break, the resulting worst-case loss stays the same as the current one.
But by definition of Delaunay triangulations such \textit{breaking} points belong to the triangle's circumcircle.
Therefore, optimal refinements must belong to this space.
Moreover, we can see that any point that leaves adjacent large triangles unrefined is dominated and we observe a set of equally optimal minimizers.
This happens because the next largest triangle's (the one highlighted in panel (b)) circumcircle does not intersect with the one of the current largest triangle.
Therefore, the refinements that bring the loss below the largest unrefinable triangle are equally optimal.
Panel (b) shows the resulting triangulation and the pointwise worst-case bias, normalized by its maximum over target values.
This maximum is attained at the center of the minimum enclosing ball of the highlighted simplex.
Note that the circle highlighted is the MEB, and not the circumcircle.

\begin{figure}[h]
  \centering
  \caption{Geometric representation of the one-step collection problem}
  \label{fig:waldron_global}
  \includegraphics[width=0.98\linewidth]{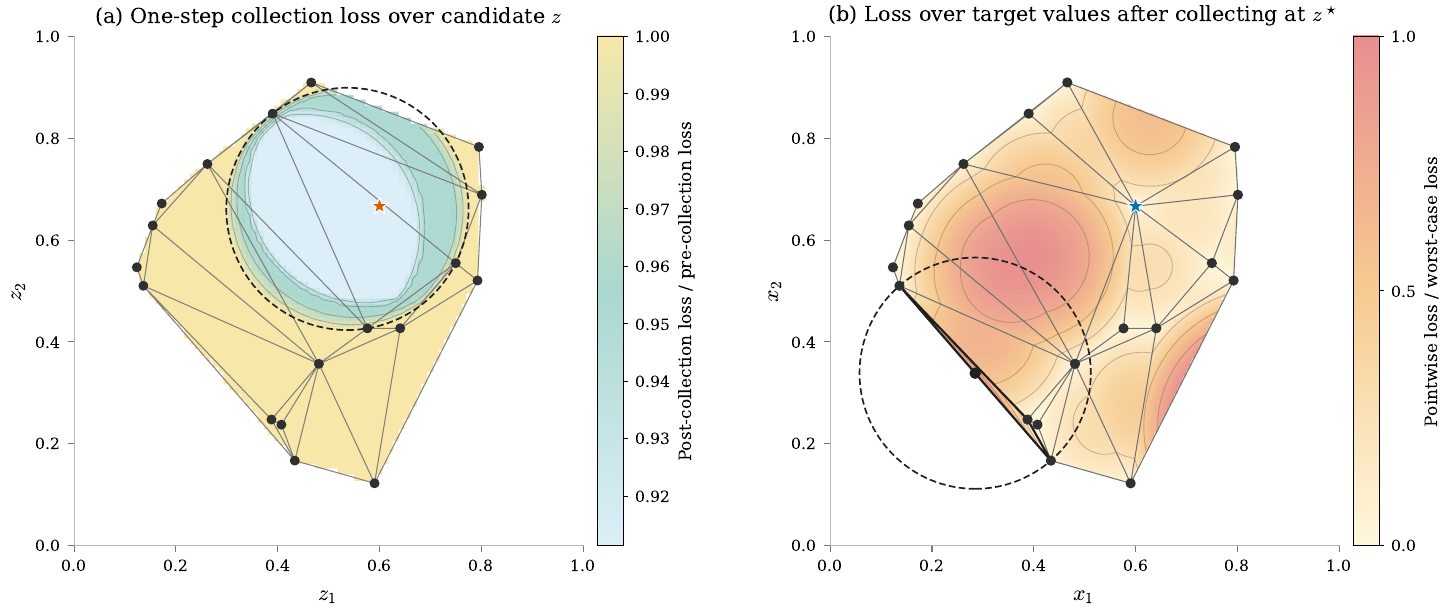}
\smallskip
\begin{minipage}{0.96\linewidth}
  \footnotesize\textit{\textbf{Notes:}} Panel (a) reports $\max_{\sigma\in T(\mathbf X^{n+1})}r_{\mathrm{meb}}(\sigma)^2$ for each candidate collection location $z$, normalized by its value before collection. The red star marks the minimizing location $z^\star$. The gray edges show the original Delaunay triangulation, and the black dashed circle is the circumsphere of the simplex containing $z^\star$. Panel (b) shows the augmented triangulation after collecting at $z^\star$, marked by the blue star. The colors report $\sum_{i=1}^{n+1}\delta_{\mathrm{del}}(\mathbf X^{n+1},x)_i\|X_i-x\|^2$, normalized by its maximum. The black outline identifies a simplex attaining the largest minimum-enclosing-ball radius, the dashed circle is its minimum enclosing circle, and the black point is its center. The normalized loss reaches one at this center. Multiplying either geometric loss by $c/2$ gives the worst-case absolute bias in Theorem~\ref{prop:collection_loss}.
\end{minipage}
\end{figure}

\section{Empirical Application}
\label{sec:application}
This section builds on the experiment of \citet{bergeron_tourek_weigel_2024} to illustrate the performance of Delaunay Matching and highlight the cost-efficiency gains of the Geometric Collection Plan. 
The authors study randomized property-tax reductions in Kananga, Democratic Republic of the Congo.
In the original field experiment, $38{,}028$ properties were randomly assigned to the status-quo annual tax liability or to a reduction of $17$, $33$, or $50$ percent. 
I suppose a researcher got access to the large sample in the status quo ($5{,}750$ properties) and could roll out the $33$-percent-reduction intervention to only $80$ out of the $5{,}909$ properties that got treated in the actual experiment.
This construction is the empirical counterpart of Example \ref{ex:coll_previous}.

The outcome is an indicator for property-tax payment. 
I characterize each property by two covariates: log tax burden, defined as the log ratio of its status-quo liability to its property value, and enforcement capacity of tax collection. 
Both covariates are provided by and used in the original analysis.
For expositional convenience, I normalize both covariates to the unit interval using their ranges in the full sample. 

I use the large samples in the status quo and in the 33-percent-reduction arm to estimate the CEFs under their respective counterfactual states.
Because the treatment was randomized, and both samples are large, this exercise credibly recovers the functions $\mu_0(x)$ and $\mu_1(x)$ (see Equation \ref{eq:cef_funct}).
I estimate each CEF with a binomial generalized additive model that uses a tensor-product thin-plate spline in the two covariates, with smoothing parameters selected by restricted maximum likelihood.
The estimated functions are then treated as known: $\mu_0(x)$ enters the treatment effect estimator as described in Definition \ref{def:del_matching_est}, while $\mu_1(x)$ is used as the ground truth to evaluate the performance of the estimators.
Figure \ref{fig:application_treatment_effect} illustrates the two CEF functions in each covariate dimension and the resulting treatment effects.

\begin{figure}[h]
  \centering
  \caption{Estimated conditional expectation and treatment-effect functions}
  \label{fig:application_treatment_effect}
  \includegraphics[width=0.7\linewidth]{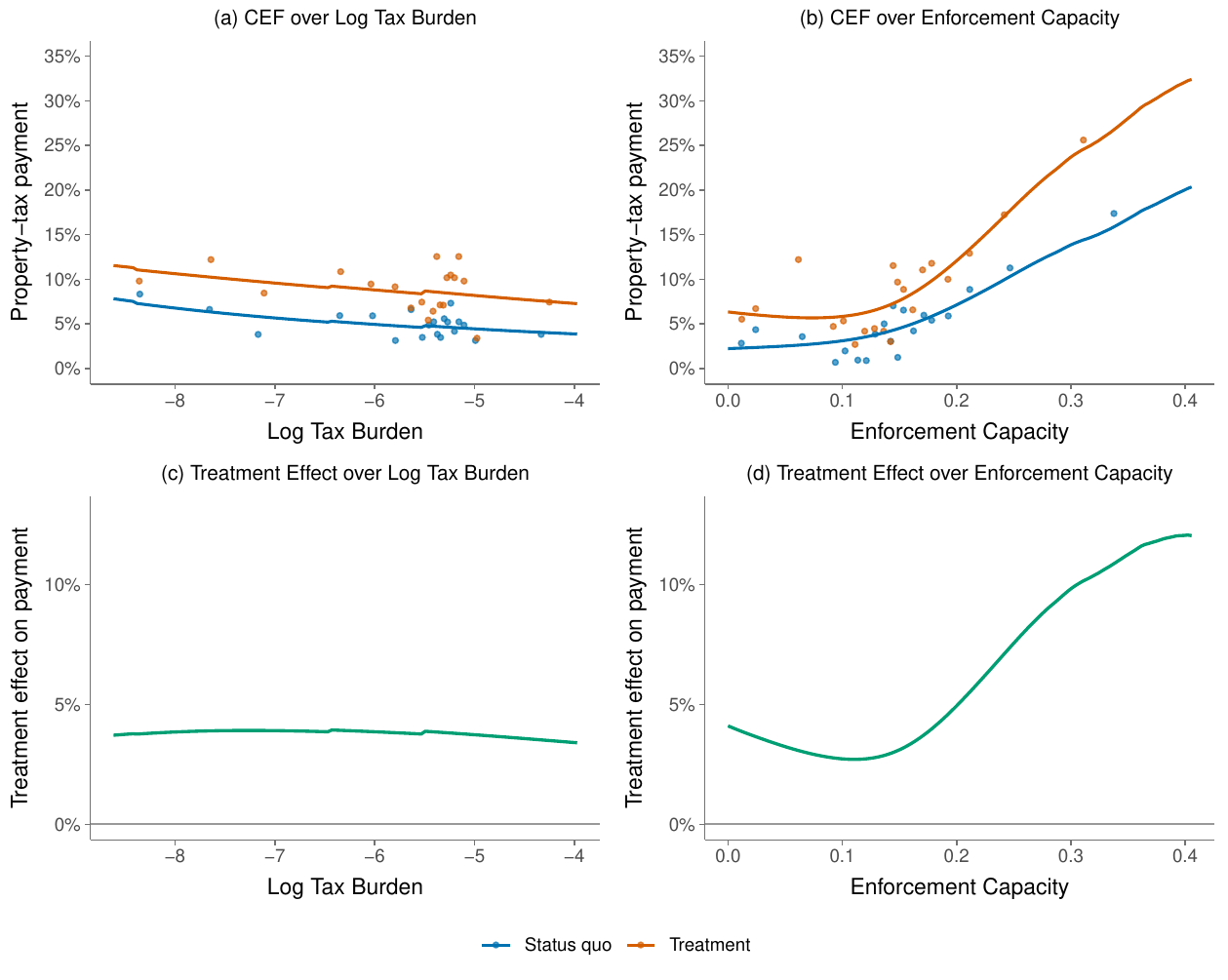}
\smallskip
\begin{minipage}{0.7\linewidth}
  \footnotesize\textit{\textbf{Notes:}} Panels (a) and (b) report the estimated status-quo and 33-percent-reduction CEFs over log tax burden and enforcement capacity, respectively; the dots are mean payment rates within 20 equal-frequency bins. Each CEF is estimated separately using a binomial generalized additive model with a logit link and a tensor-product thin-plate spline in the two covariates. The basis dimension is six in each covariate, and smoothing is selected by restricted maximum likelihood with an additional selection penalty and $\gamma=2$. To construct each one-dimensional profile, I vary the covariate on the horizontal axis over the common first-to-99th-percentile range, average predictions over the pooled empirical distribution of the other covariate, and retain only covariate combinations within the convex hull of both experimental groups. Panels (c) and (d) report the corresponding treatment-effect functions, computed as the average difference between the two observation-level CEF predictions.
\end{minipage}
\end{figure}

To reproduce the asymmetric sampling environment allowed by Assumption~\ref{ass:samp_status_quo} and consistent with Example \ref{ex:coll_previous}, I draw without replacement a random sample of $n=80$ properties from the original treatment group and regard it as the finite intervention sample $S^n$. 
All subsequent matching and collection results hold this initial draw fixed.

Figure~\ref{fig:application_design_support} describes the resulting support. Panel (a) plots the $80$ sampled intervention units in black, the remaining intervention units in gray, and the status-quo targets according to whether they lie inside the convex hull of the collection design. Of the $5{,}750$ targets, $4{,}931$ lie inside this hull and are therefore feasible for convex matching; the remaining $819$ are excluded from the analysis. Panel (b) shows the Delaunay triangulation of the sampled intervention units together with the feasible targets. The triangulation partitions the feasible region into simplices and therefore determines which intervention observations receive positive weight at each target value.

\begin{figure}[h]
  \centering
  \caption{Intervention collection design and target support}
  \label{fig:application_design_support}
  \includegraphics[width=0.98\linewidth]{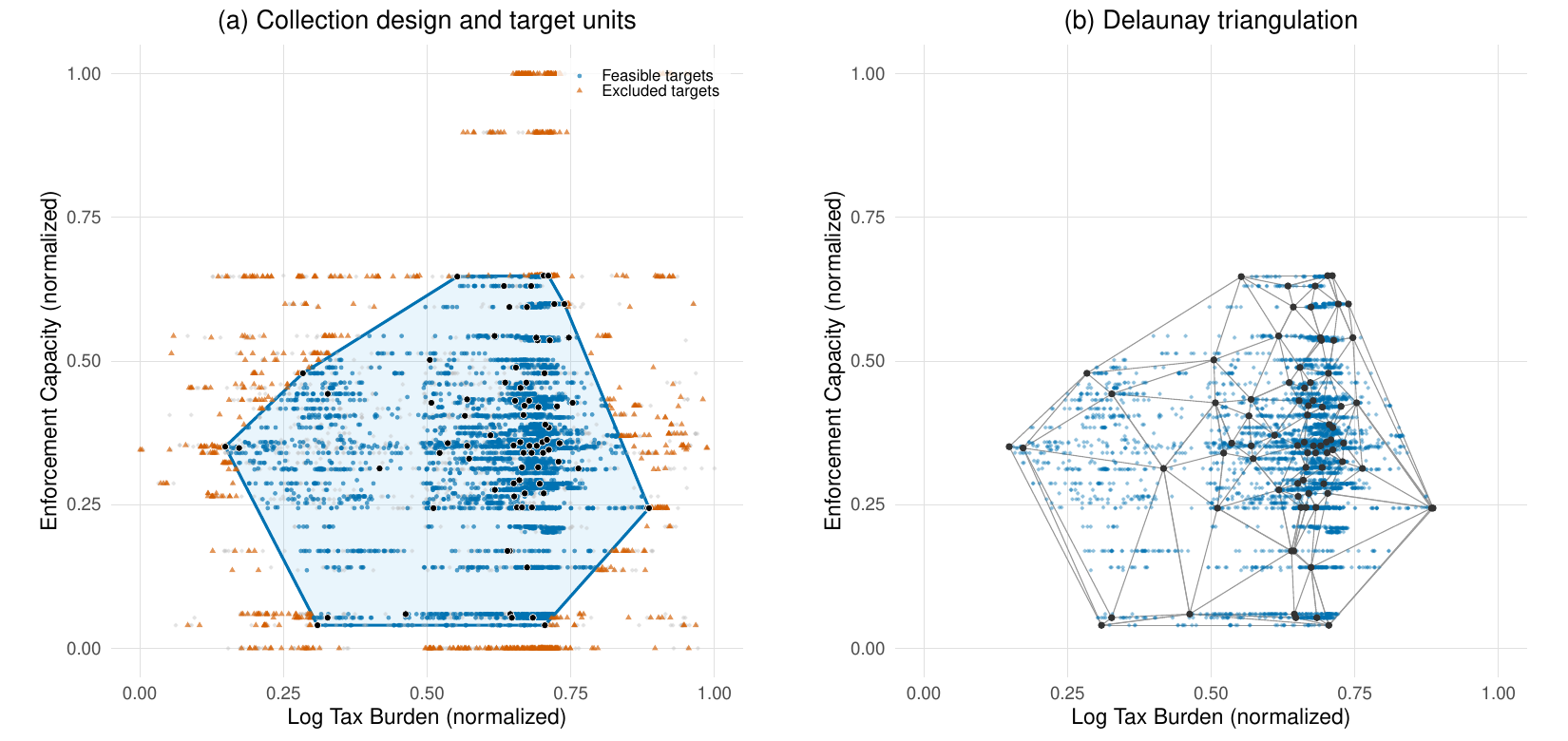}
\smallskip
\begin{minipage}{0.96\linewidth}
  \footnotesize\textit{\textbf{Notes:}} Panel (a) reports the sampled intervention units, the remaining intervention units, and the status-quo targets relative to the convex hull of the collection design. Panel (b) reports the Delaunay triangulation of the sampled intervention units and the feasible status-quo targets.
\end{minipage}
\end{figure}

\paragraph{Delaunay Matching Performance.}

I apply Delaunay Matching to every target value in the convex hull shown in Figure~\ref{fig:application_design_support}. For a target $x$, panel (b) determines the Delaunay triangle $\sigma_n(x)$ that contains it. 
Conditional on the realized collection design, the estimator's absolute interpolation error under the estimated intervention CEF is
\begin{equation}
    \left|
    \widehat\mu_1(x)-
    \sum_{i = 1}^n
    \delta_\mathrm{del}(\mathbf{X}^n,x)_i \widehat\mu_1(X_i)
    \right|.
\end{equation}
By contrast, as shown in Proposition \ref{prop:matching_problem}, the worst-case bias over $\mathcal{F}_c$ is
\begin{equation}
    \frac{c}{2}
    \sum_{i = 1}^n
    \delta_\mathrm{del}(\mathbf{X}^n,x)_i\|X_i-x\|^2,
\end{equation}
where $c$ is a uniform upper bound on the operator norm of the Hessian in the normalized covariate coordinates. 
I calibrate $c$ numerically from $\widehat\mu_1$: finite-difference Hessians evaluated over the collection-design convex hull yield a maximum operator norm of $3.411$.
Note that this step is only necessary for illustration purposes; it is not necessary in practice to know $c$ to conduct Delaunay Matching or the Geometric Collection Plan. 
Figure~\ref{fig:application_delaunay_bias} compares these two notions of bias: in Panel (a) the former and in Panel (b) the latter.
Both quantities vanish on the elements of the collection design and increase within large simplices. 
The worst-case bound is conservative but informative: over the displayed grid, the interpolation error reaches at most $0.0205$, whereas the worst-case bound reaches $0.0667$.

\begin{figure}[h]
  \centering
  \caption{Pointwise Delaunay Matching bias}
  \label{fig:application_delaunay_bias}
  \includegraphics[width=0.98\linewidth]{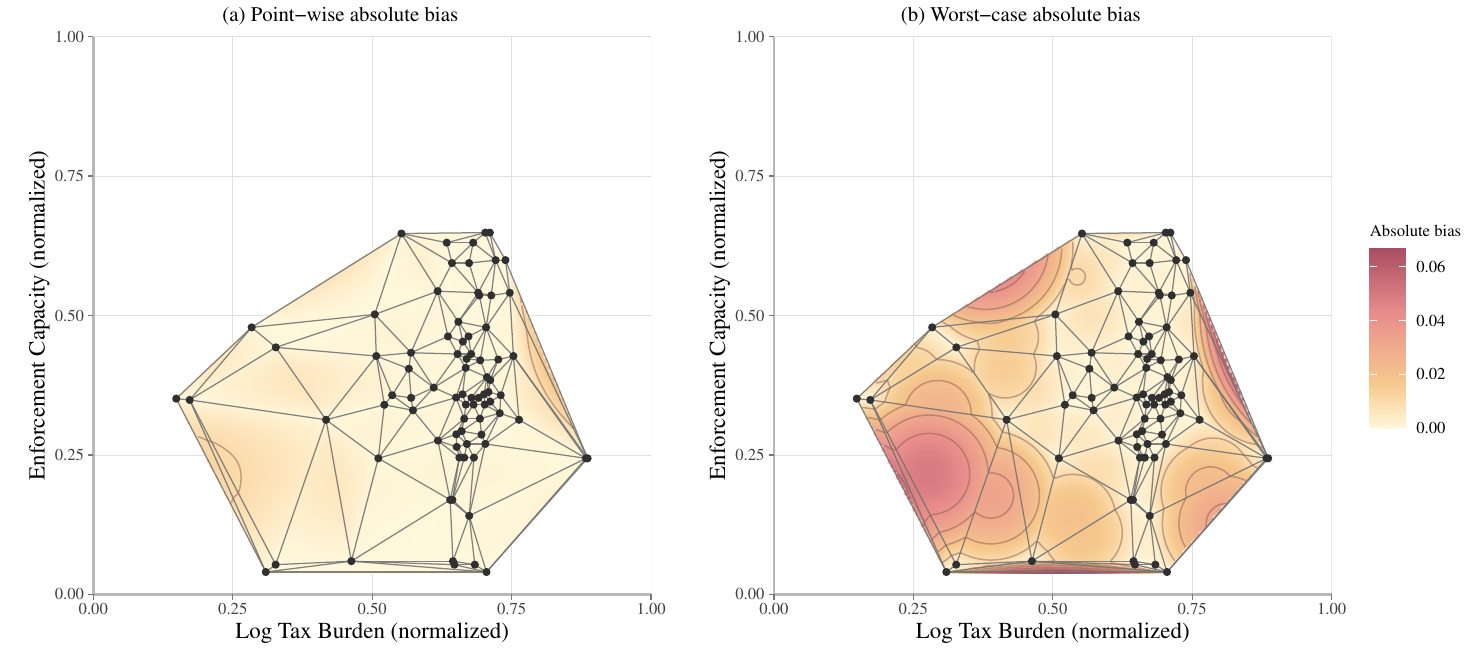}
\smallskip
\begin{minipage}{0.96\linewidth}
  \footnotesize\textit{\textbf{Notes:}} Panel (a) reports the absolute interpolation error of Delaunay Matching under the estimated intervention CEF. Panel (b) reports the pointwise worst-case absolute bias over the bounded-curvature class $\mathcal F_c$, using the numerically calibrated value $c=3.411$. Black points are the $80$ intervention units and gray edges are their Delaunay triangulation. Both panels use the same color scale.
\end{minipage}
\end{figure}

I compare Delaunay Matching with three alternative convex, affine-exact estimators. All three choose weighting rules in $\Delta$.
For the \emph{Minimum-Norm} benchmark, I consider:
\begin{equation}
    \delta_{\mathrm{MN}}\in\arg\min_{\delta\in\Delta}\sum_{i=1}^n\delta(\mathbf X^n,x)_i^2.
\end{equation}
\emph{Entropy Balancing} chooses the feasible vector with maximum Shannon entropy:
\begin{equation}
    \delta_{\mathrm{EB}}\in\arg\max_{\delta\in\Delta}\left\{-\sum_{i=1}^n\delta(\mathbf X^n,x)_i\log\delta(\mathbf X^n,x)_i\right\},
    \qquad 0\log0:=0.
\end{equation}
\emph{Adaptive Nearest Neighbor} orders donors by their Euclidean distance from $x$, breaking ties by donor index. Let $N_q(x)$ contain the $q$ closest donors and define
\begin{equation}
    q^\star(x):=\min\left\{q\in\{3,\ldots,n\}:x\in\operatorname{conv}\{X_i:i\in N_q(x)\}\right\}.
\end{equation}
In this two-dimensional application, this gives the smallest feasible neighborhood of at least three donors. The estimator computes minimum-norm exact-balancing weights within that neighborhood:
\begin{equation}
    \delta_{\mathrm{ANN}}\in\arg\min_{\substack{\delta\in\Delta\\i\in N_{q^\star(x)}}}\sum_{i=1}^n\delta(\mathbf X^n,x)_i^2,
\end{equation}
and $\delta(\mathbf X^n,x)_i=0$ for $i\notin N_{q^\star(x)}$.
Because every estimator satisfies the restrictions defining $\Delta$, by Proposition \ref{prop:matching_problem}, its pointwise worst-case bias is proportional to
\begin{equation}
    \sum_{i=1}^n\delta_{m}(\mathbf{X}^n,x)_i\|X_i-x\|^2.
\end{equation}
I evaluate the loss at $864$ points on a regular grid inside the donor convex hull. 
At each target, I divide an estimator's loss by the loss attained by Delaunay Matching, subtract one, and multiply by $100$ to obtain its percentage increase in worst-case bias relative to DME.

Figure~\ref{fig:application_estimator_comparison} reports the distribution over the covariate values of the pointwise percentage increase relative to Delaunay Matching. 
The ranking is consistent with Theorem~\ref{thm:delaunay_budget}: none of the competing weighting rules improves upon Delaunay at any target. 
As expected, the diffuse weighting schemes perform especially poorly in terms of bias. 
The median worst-case bias is approximately $426\%$ larger for Synthetic Control and $461\%$ larger for Entropy Balancing. 
Restricting balance to a local feasible neighborhood substantially improves performance: the median increase for Adaptive Nearest Neighbor is approximately $24\%$. 
Its upper tail nevertheless shows that selecting donors solely by proximity need not recover the Delaunay Matching rule.

\begin{figure}[h]
  \centering
  \caption{Pointwise comparison with alternative matching estimators}
  \label{fig:application_estimator_comparison}
  \includegraphics[width=0.68\linewidth]{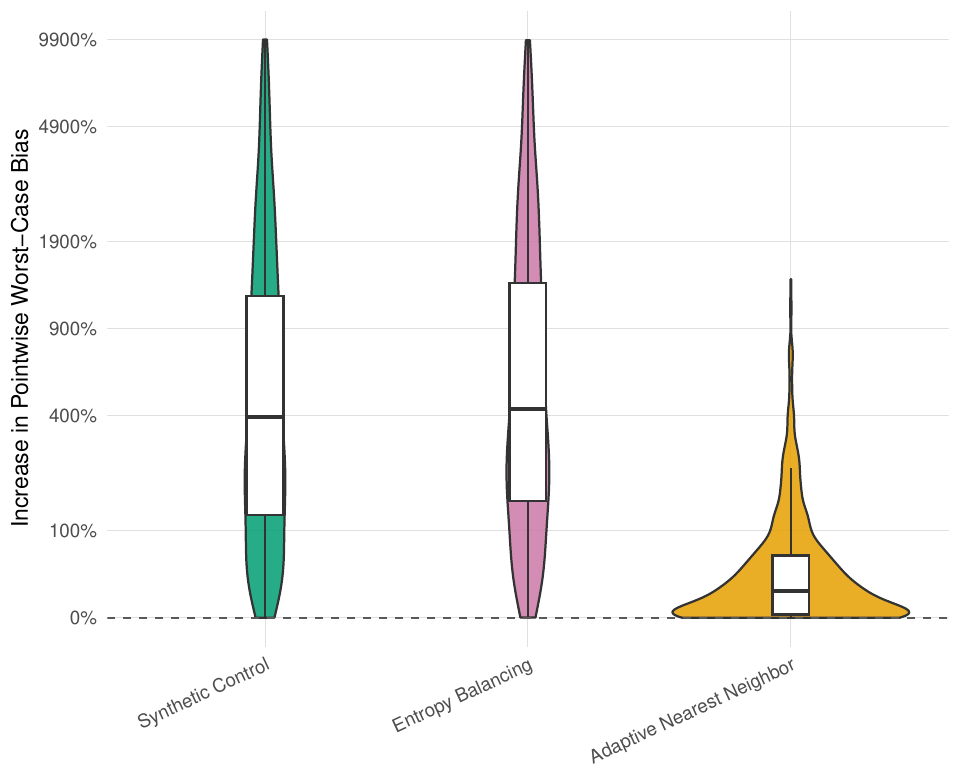}
\smallskip
\begin{minipage}{0.82\linewidth}
  \footnotesize\textit{\textbf{Notes:}} For each grid point and estimator, the figure reports $100\times[\sum_{i=1}^n\delta_{m}(\mathbf{X}_n,x)_i\|X_i-x\|^2/\sum_{i=1}^n\delta_{\mathrm{del}}(\mathbf{X}_n,x)_i\|X_i-x\|^2-1]$. Violins show the distribution across the $864$ interior grid points and boxes report the interquartile range and median. Delaunay Matching is omitted because its increase is identically zero. The vertical axis is logarithmic; for readability, pointwise ratios above $100$---increases above $9{,}900\%$---are omitted from the displayed distributions.
\end{minipage}
\end{figure}

\paragraph{Geometric Collection Performance.}

Let's now consider where additional intervention outcomes should be collected. 
In the application, this means which new properties should be assigned the tax reduction next to decrease the largest worst-case bias we suffer over the target values as much as possible.
The initial collection design is the set of $80$ black locations in Figure~\ref{fig:application_design_support}. I hold its convex hull fixed and add one donor location at a time within that region. After each addition, the Delaunay triangulation is recomputed. 

I compare three sequential rules. The \emph{minimax plan} numerically solves the one-step problem in Definition~\ref{def:collection_plan} after every addition, choosing the location that minimizes $r_\mathrm{meb}(\sigma^*)$ for the refined design. 
\emph{Random sampling} instead draws each new location uniformly from the initial convex hull. 
Finally, \emph{random simplex refinement} draws a triangle from the current Delaunay triangulation with probability proportional to its area and collects at the center of that triangle's minimum enclosing ball. 

Figure~\ref{fig:application_collection_paths} illustrates the first five additions. Because the Geometric Plan effectively knows \textit{where to look} for optimal additions, it places them in regions that determine the global collection loss. 
By contrast, uniform random sampling need not refine consequential parts of the triangulation. 
Random simplex refinement lies between these two approaches: the simplex (highlighted in purple in panel (c)) is selected at random, but its added point is chosen to reduce its geometric interpolation gap. 

\begin{figure}[h]
  \centering
  \caption{Five-step collection paths}
  \label{fig:application_collection_paths}
  \includegraphics[width=0.99\linewidth]{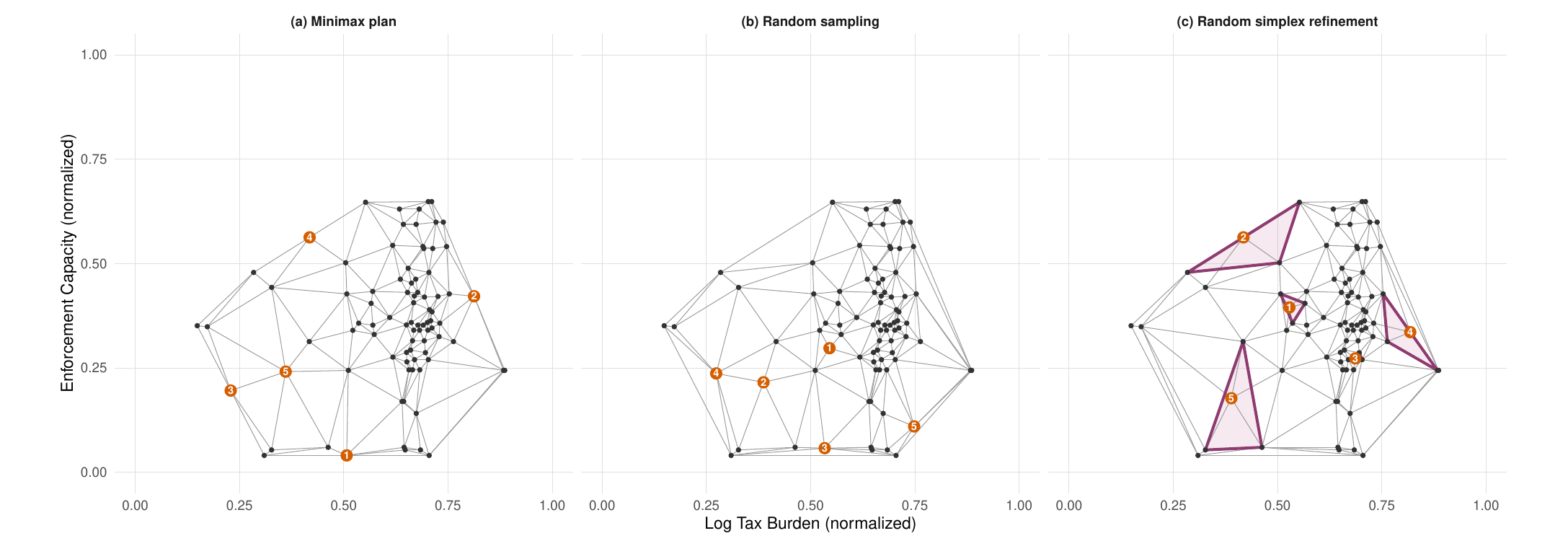}
\smallskip
\begin{minipage}{0.97\linewidth}
  \footnotesize\textit{\textbf{Notes:}} Each panel starts from the same $80$ intervention donors. Orange numbered points report the locations and order of the first five additions; gray edges show the Delaunay triangulation after the fifth addition. The middle and right panels show paths whose five-step collection losses are closest to the respective medians across $100$ simulations. In the right panel, each purple triangle is the simplex selected at the indicated step, drawn using the triangulation that existed immediately before that point was added; it may therefore be absent from the final triangulation.
\end{minipage}
\end{figure}

Figure~\ref{fig:application_collection_budget} compares the reduction in worst-case absolute bias as the collection budget increases. 
The advantage of solving the global geometric problem is large even with few additions. 
The first minimax addition reduces the criterion by approximately $5.5\%$, four additions reduce it by $50.8\%$, and twenty additions reduce it by $82.4\%$. 
If each intervention outcome has the same collection cost, these gains correspond to average reductions in the initial worst-case bias of $5.5$, $12.7$ $(= 50.8/4)$, and $4.1$ $(= 82.4/20)$ percentage points per observation at budgets of one, four, and twenty observations, respectively. 
These averages need not decline monotonically because the global criterion is determined by the worst remaining simplex, making the returns to successive additions lumpy and potentially complementary.
By contrast, none of the $100$ uniform-random paths changes the global worst-case bias within twenty additions. 
Random simplex refinement occasionally makes a consequential addition, but its median reduction also remains zero through twenty points: only $24$ of the $100$ random paths improve the global criterion. 
In particular, if we focus on \textit{lucky draws} (i.e., the top 1\% of the random draws) of the random simplex refinement, we see that 10 additions achieve the same worst-case bias improvement as 2 under the Geometric Plan. 
Therefore, we conclude that the Geometric Plan leads to sizeable gains in absolute and per-unit-cost performance.

\begin{figure}[h!]
  \centering
  \caption{Collection performance by budget}
  \label{fig:application_collection_budget}
  \includegraphics[width=0.76\linewidth]{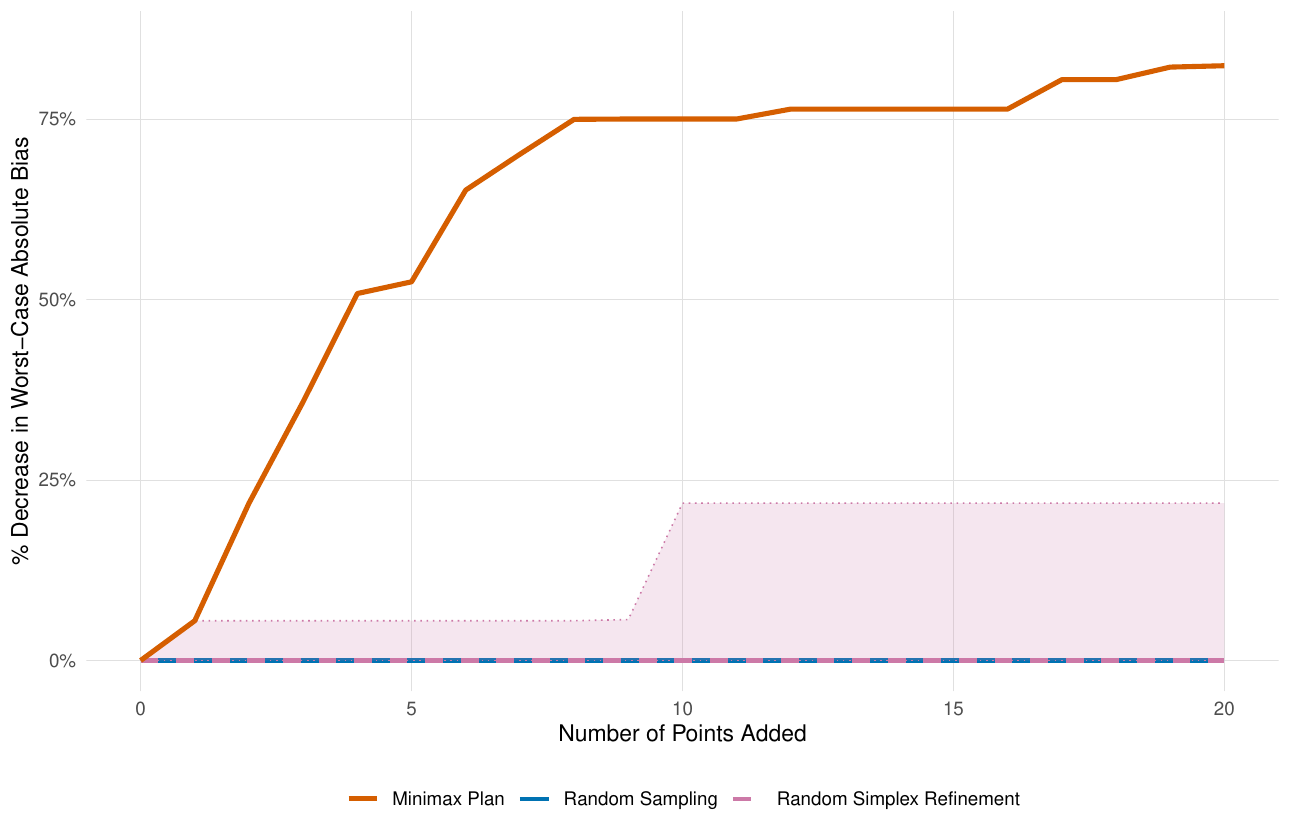}
\smallskip
\begin{minipage}{0.84\linewidth}
  \footnotesize\textit{\textbf{Notes:}} The figure reports $100\times[1-\sum_{i=1}^n\delta_{m}(\mathbf{X}_{80 + b},x)_i\|X_i-x\|^2/\sum_{i=1}^n\delta_{\mathrm{del}}(\mathbf{X}_{80},x)_i\|X_i-x\|^2]$ after adding $b$ intervention observations. Because worst-case absolute bias is proportional to it, this is also its percentage reduction. The orange line reports the sequential minimax plan. Blue and purple lines report medians across $100$ random paths, and their shaded regions span the 1st to 99th percentiles. Random sampling draws locations uniformly from the initial convex hull. Random simplex refinement selects a current Delaunay triangle with probability proportional to its area and adds the center of its minimum enclosing ball.
\end{minipage}
\end{figure}

\section{Conclusions}
\label{sec:conclusions}

This paper connects the choice of matching rules and the collection of additional intervention data through the geometry of the collection design.
Under bounded curvature, Delaunay Matching minimizes worst-case absolute conditional bias among convex, affine-exact estimators at every target in the collection design convex hull.
The same triangulation provides an exact representation of the one-step collection objective and restricts where an improving addition can lie, effectively helping us \textit{learn where to look}.
Neither choice requires knowing the magnitude of the curvature bound.
The empirical application illustrates how these results can guide estimation and collection in practice.
Interesting avenues for future research include extending the results to treatment choice and evidence aggregation, or allowing the two samples to belong to different populations.

\bibliographystyle{plainnat}
\bibliography{bib}

\clearpage
\appendix
\section{Proofs}
\doublespace
\label{sec:proofs}
\begin{proof}[Proof of Proposition \ref{prop:matching_problem}]
Note that
\begin{align}
    \tau(x) - \hat{\tau}_m(x) 
    & = \mu_1(x) - \mu_0(x) - \hat{\mu}_{1,m}(x) + \mu_0(x) \\
    & = \mu_1(x) - \hat{\mu}_{1,m}(x).
\end{align}
Assumption \ref{ass:cef_int} implies that
\begin{equation}
    \mathbb{E}_P[Y_i(1) | X_i = x, D_i = d] = f(x),
\end{equation}
and therefore 
\begin{equation}
    \mu_1(x) = f(x).
\end{equation}
Therefore,
\begin{align}
    \E_{P_1^n}[\tau(x) - \hat{\tau}_m(x) | \mathbf{X}^n] 
    & = \E_{P_1^n}[ \mu_1(x) - \hat{\mu}_{1,m}(x) | \mathbf{X}^n]  \\
    & = \E_{P_1^n}[ f(x) - \hat{\mu}_{1,m}(x) | \mathbf{X}^n].
\end{align}
But then
\begin{equation}
    \sup_{P \in \mathcal{P}_c}|\E_{P_1^n}[\tau(x) - \hat{\tau}_m(x) | \mathbf{X}^n]| = \max_{f \in \mathcal{F}_c} |\E_{P_1^n}[ f(x) - \hat{\mu}_{1,m}(x) | \mathbf{X}^n]|
\end{equation}
Now, define
\begin{equation}
    \hat{\mu}_{1,\delta}(x):= \sum_{i=1}^n \delta(\mathbf{X}^n, x)_i Y_i.
\end{equation}
Then,
\begin{equation}
    \E_{P_1^n}[ f(x) - \hat{\mu}_{1,m}(x) | \mathbf{X}^n] = f(x) - \E_{P_1^n}\left[\sum_{i = 1}^n \delta(\mathbf{X}^n, x)_i Y_i \ \middle|\ \mathbf{X}^n\right]
\end{equation}
Because the expectation is taken conditional on $\mathbf{X}^n$, and the weights depend only on $\mathbf{X}^n$ and $x$, 
\begin{align}
    f(x) & - \E_{P_1^n}\left[\sum_{i = 1}^n \delta(\mathbf{X}^n, x)_i Y_i \ \middle|\ \mathbf{X}^n\right] = \\
    & = f(x) - \sum_{i = 1}^n \delta(\mathbf{X}^n, x)_i \cdot \E_{P_1^n}\left[ Y_i \middle|\mathbf{X}^n\right] \\
    & = f(x) - \sum_{i = 1}^n \delta(\mathbf{X}^n, x)_i \cdot \E_{P_1^n}[Y_i(1)| \mathbf{X}^n, D_i = 1 ] \\
    & = f(x) - \sum_{i = 1}^n \delta(\mathbf{X}^n, x)_i \cdot \E_{P_1^n}[Y_i(1)| \{X_1, ..., X_n\}, D_i = 1 ] \\
    & = f(x) - \sum_{i = 1}^n \delta(\mathbf{X}^n, x)_i \cdot \E_{P_1^n}[Y_i(1)| X_i, D_i = 1 ] \\
    & = f(x) - \sum_{i = 1}^n \delta(\mathbf{X}^n, x)_i \cdot f(X_i).
\end{align}
Therefore,
\begin{equation}
    \min_{\delta \in \Delta} \max_{f \in \mathcal{F}_c}  \left|\E_{P_1^n}\left[ \tau(x) - \hat{\tau}_\delta(x)\ \middle|\ \mathbf{X}^n\right]\right| = \min_{\delta \in \Delta} \max_{f \in \mathcal{F}_c} \left| f(x) - \sum_{i = 1}^n \delta(\mathbf{X}^n, x)_i \cdot f(X_i) \right|.
\end{equation}

It remains to evaluate the inner maximization. Fix any $\delta\in\Delta$.
For every $f\in\mathcal F_c$, Taylor's theorem with integral remainder gives, for each $i$,
\begin{align}
f(X_i)
&=
f(x)
+\nabla f(x)^\top(X_i-x) \\
&\quad+
\int_0^1(1-t)(X_i-x)^\top
\nabla^2 f\bigl(x+t(X_i-x)\bigr)(X_i-x)\,dt.
\end{align}
Multiplying by $\delta(\mathbf X^n,x)_i$, summing over $i$, and subtracting $f(x)$ yields
\begin{align}
&\sum_{i=1}^n\delta(\mathbf X^n,x)_i f(X_i)-f(x) \\
&\quad=
f(x)\left(\sum_{i=1}^n\delta(\mathbf X^n,x)_i-1\right)
+\nabla f(x)^\top
\left(\sum_{i=1}^n\delta(\mathbf X^n,x)_iX_i-x\right) \\
&\qquad+
\sum_{i=1}^n\delta(\mathbf X^n,x)_i
\int_0^1(1-t)(X_i-x)^\top
\nabla^2 f\bigl(x+t(X_i-x)\bigr)(X_i-x)\,dt.
\end{align}
By the two restrictions in the definition of $\Delta$, the first two terms on the right-hand side are zero. Hence
\begin{align}
&\left|f(x)-\sum_{i=1}^n\delta(\mathbf X^n,x)_i f(X_i)\right| \\
&\quad\leq
\sum_{i=1}^n\delta(\mathbf X^n,x)_i
\int_0^1(1-t)
\left|(X_i-x)^\top
\nabla^2 f\bigl(x+t(X_i-x)\bigr)(X_i-x)\right|dt \\
&\quad\leq
\sum_{i=1}^n\delta(\mathbf X^n,x)_i
\int_0^1(1-t)c\|X_i-x\|^2dt \\
&\quad=
\frac c2\sum_{i=1}^n
\delta(\mathbf X^n,x)_i\|X_i-x\|^2.
\end{align}
The first inequality uses the nonnegativity of the weights, and the second follows from the operator-norm bound in the definition of $\mathcal F_c$. Therefore,
\begin{equation}
\max_{f\in\mathcal F_c}
\left|f(x)-\sum_{i=1}^n\delta(\mathbf X^n,x)_i f(X_i)\right|
\leq
\frac c2\sum_{i=1}^n
\delta(\mathbf X^n,x)_i\|X_i-x\|^2.
\end{equation}

To establish the reverse inequality, consider the function
\begin{equation}
f(\cdot)=\frac c2\|\cdot\|^2.
\end{equation}
Its Hessian has operator norm $c$, so it belongs to $\mathcal F_c$. Moreover, the restrictions defining $\Delta$ imply
\begin{align}
\sum_{i=1}^n\delta(\mathbf X^n,x)_i\|X_i-x\|^2
&=
\sum_{i=1}^n\delta(\mathbf X^n,x)_i\|X_i\|^2
-2x^\top\sum_{i=1}^n\delta(\mathbf X^n,x)_iX_i
+\|x\|^2\sum_{i=1}^n\delta(\mathbf X^n,x)_i \\
&=
\sum_{i=1}^n\delta(\mathbf X^n,x)_i\|X_i\|^2-\|x\|^2.
\end{align}
Consequently, for this function,
\begin{align}
\left|f(x)-\sum_{i=1}^n\delta(\mathbf X^n,x)_i f(X_i)\right|
&=
\frac c2
\left|\|x\|^2-
\sum_{i=1}^n\delta(\mathbf X^n,x)_i\|X_i\|^2\right| \\
&=
\frac c2\sum_{i=1}^n
\delta(\mathbf X^n,x)_i\|X_i-x\|^2.
\end{align}
Thus the upper bound is attained, and, for every $\delta\in\Delta$,
\begin{equation}
\max_{f\in\mathcal F_c}
\left|f(x)-\sum_{i=1}^n\delta(\mathbf X^n,x)_i f(X_i)\right|
=
\frac c2\sum_{i=1}^n
\delta(\mathbf X^n,x)_i\|X_i-x\|^2.
\end{equation}
Taking the minimum over $\delta\in\Delta$ on both sides and combining this identity with the preceding reduction proves the result.
\end{proof}
\clearpage

\begin{proof}[Proof of Theorem \ref{thm:delaunay_budget}]
First, $\delta_\mathrm{del}\in\Delta$. Indeed, by definition,
\begin{equation}
\delta_\mathrm{del}(\mathbf X^n,x)_i\geq0
\end{equation}
for every $i$. Moreover, because the weights are zero whenever $X_i$ is not a vertex of $\sigma_n(x)$, the restrictions satisfied by $\lambda$ imply
\begin{align}
\sum_{i=1}^n\delta_\mathrm{del}(\mathbf X^n,x)_i
&=
\sum_{\{i\in \nu_n(x)\}}
\lambda(\sigma_n(x),x)_i
=1, \\
\sum_{i=1}^n\delta_\mathrm{del}(\mathbf X^n,x)_iX_i
&=
\sum_{\{i\in \nu_n(x)\}}
\lambda(\sigma_n(x),x)_iX_i
=x.
\end{align}

Now, fix any $\delta\in\Delta$. The restrictions in the definition of $\Delta$ give
\begin{align}
\sum_{i=1}^n
\delta(\mathbf X^n,x)_i\|X_i-x\|^2
&=
\sum_{i=1}^n
\delta(\mathbf X^n,x)_i
\left(\|X_i\|^2-2x^\top X_i+\|x\|^2\right) \\
&=
\sum_{i=1}^n
\delta(\mathbf X^n,x)_i\|X_i\|^2
-2x^\top\sum_{i=1}^n
\delta(\mathbf X^n,x)_iX_i
+\|x\|^2\sum_{i=1}^n
\delta(\mathbf X^n,x)_i \\
&=
\sum_{i=1}^n
\delta(\mathbf X^n,x)_i\|X_i\|^2-\|x\|^2.
\end{align}
The same identity applies to $\delta_\mathrm{del}$ because $\delta_\mathrm{del}\in\Delta$. Therefore, it is sufficient to prove that
\begin{equation}
\sum_{i=1}^n
\delta_\mathrm{del}(\mathbf X^n,x)_i\|X_i\|^2
\leq
\sum_{i=1}^n
\delta(\mathbf X^n,x)_i\|X_i\|^2.
\end{equation}

Let $a$ and $r$ be the center and the radius of the circumsphere of $\sigma_n(x)$. Because $\sigma_n(x)$ is a simplex of the Delaunay triangulation, the interior of its circumsphere contains no donor point. Consequently,
\begin{equation}
\|X_i-a\|^2\geq r^2
\end{equation}
for every $i$, with equality whenever $X_i$ is a vertex of $\sigma_n(x)$. Expanding the squared norm gives
\begin{equation}
\|X_i\|^2
\geq
2a^\top X_i+r^2-\|a\|^2
\end{equation}
for every $i$. Multiplying this inequality by the nonnegative weight $\delta(\mathbf X^n,x)_i$ and summing over $i$ yields
\begin{align}
\sum_{i=1}^n
\delta(\mathbf X^n,x)_i\|X_i\|^2
&\geq
2a^\top\sum_{i=1}^n
\delta(\mathbf X^n,x)_iX_i
+(r^2-\|a\|^2)
\sum_{i=1}^n\delta(\mathbf X^n,x)_i \\
&=
2a^\top x+r^2-\|a\|^2.
\end{align}

On the other hand, $\delta_\mathrm{del}(\mathbf X^n,x)_i$ is positive only if $X_i$ is a vertex of $\sigma_n(x)$. The circumsphere inequality is therefore an equality for every donor that receives positive Delaunay weight. Hence
\begin{align}
\sum_{i=1}^n
\delta_\mathrm{del}(\mathbf X^n,x)_i\|X_i\|^2
&=
2a^\top\sum_{i=1}^n
\delta_\mathrm{del}(\mathbf X^n,x)_iX_i
+(r^2-\|a\|^2)
\sum_{i=1}^n\delta_\mathrm{del}(\mathbf X^n,x)_i \\
&=
2a^\top x+r^2-\|a\|^2.
\end{align}
Combining the last two displays gives
\begin{equation}
\sum_{i=1}^n
\delta_\mathrm{del}(\mathbf X^n,x)_i\|X_i\|^2
\leq
\sum_{i=1}^n
\delta(\mathbf X^n,x)_i\|X_i\|^2.
\end{equation}
Using the squared-distance identity established above on both sides, it follows that
\begin{equation}
\sum_{i=1}^n
\delta_\mathrm{del}(\mathbf X^n,x)_i\|X_i-x\|^2
\leq
\sum_{i=1}^n
\delta(\mathbf X^n,x)_i\|X_i-x\|^2.
\end{equation}
Because $\delta\in\Delta$ was arbitrary, $\delta_\mathrm{del}$ minimizes the right-hand side of Proposition \ref{prop:matching_problem}. The result then follows directly from Proposition \ref{prop:matching_problem}.
\end{proof}
\clearpage

\begin{proof}[Proof of Theorem \ref{prop:collection_loss}]
Fix $z\in\operatorname{conv}(\mathbf X^n)$. Because adding $z$ does not change the convex hull,
\begin{equation}
\operatorname{conv}(\mathbf X^{n+1})
=
\operatorname{conv}(\mathbf X^n).
\end{equation}
For every $x\in\operatorname{conv}(\mathbf X^n)$, the argument in the proof of Proposition \ref{prop:matching_problem}, applied to the augmented design and to the Delaunay weights, gives
\begin{align}
&\sup_{P\in\mathcal P_c}
\left|\E_{P_1^{n+1}}
\left[\tau(x)-\hat\tau_{\mathrm{del}}(x)\mid\mathbf X^{n+1}\right]\right| \\
&\qquad=
\frac c2\sum_{i=1}^{n+1}
\delta_{\mathrm{del}}(\mathbf X^{n+1},x)_i\|X_i-x\|^2.
\end{align}

Now fix a simplex $\sigma\in T(\mathbf X^{n+1})$ and a point $x\in\sigma$. By the definition of Delaunay matching, only the vertices of $\sigma$ receive positive weight, and those weights are the barycentric coordinates of $x$. If $x$ belongs to more than one simplex, any two such simplices intersect in a common face containing $x$; the barycentric weights on the vertices outside that face are zero, so the value below does not depend on which containing simplex is selected. Hence
\begin{equation}
\sum_{i=1}^{n+1}
\delta_{\mathrm{del}}(\mathbf X^{n+1},x)_i\|X_i-x\|^2
=
\sum_{\{i:X_i\text{ is a vertex of }\sigma\}}
\lambda(\sigma,x)_i\|X_i-x\|^2.
\end{equation}
Let $a$ be the center of the minimum enclosing ball of the vertices of $\sigma$, and let $r_{\mathrm{meb}}(\sigma)$ be its radius. For every vertex $X_i$ of $\sigma$,
\begin{equation}
\|X_i-a\|^2\leq r_{\mathrm{meb}}(\sigma)^2.
\end{equation}
Since the barycentric weights sum to one and their weighted average of the vertices is $x$, the variance decomposition gives
\begin{align}
&\sum_{\{i:X_i\text{ is a vertex of }\sigma\}}
\lambda(\sigma,x)_i\|X_i-x\|^2 \\
&\qquad=
\sum_{\{i:X_i\text{ is a vertex of }\sigma\}}
\lambda(\sigma,x)_i\|X_i-a\|^2-
\|x-a\|^2 \\
&\qquad\leq r_{\mathrm{meb}}(\sigma)^2.
\end{align}
Therefore,
\begin{equation}
\sup_{x\in\sigma}
\sum_{i=1}^{n+1}
\delta_{\mathrm{del}}(\mathbf X^{n+1},x)_i\|X_i-x\|^2
\leq r_{\mathrm{meb}}(\sigma)^2.
\end{equation}

To establish the reverse inequality, first note that $a$ belongs to the convex hull of the vertices that lie on the boundary of the minimum enclosing ball. Suppose otherwise. By strict separation, there would exist a direction along which $a$ could be moved so as to strictly reduce its distance from every boundary vertex. For a sufficiently small movement in that direction, the distances from the remaining vertices would still be strictly smaller than $r_{\mathrm{meb}}(\sigma)$, because there are finitely many vertices and each of them has positive slack. The moved center would therefore define a ball containing every vertex of $\sigma$ with radius strictly smaller than $r_{\mathrm{meb}}(\sigma)$, contradicting the definition of the minimum enclosing ball.

It follows that $a\in\sigma$. Moreover, because the vertices of a simplex are affinely independent, the barycentric representation of $a$ is unique. The representation of $a$ as a convex combination of boundary vertices must therefore coincide with its barycentric representation in $\sigma$, so the barycentric coordinates of $a$ are zero on every vertex lying strictly inside the minimum enclosing ball. Since every vertex receiving positive barycentric weight lies on its boundary,
\begin{align}
&\sum_{\{i:X_i\text{ is a vertex of }\sigma\}}
\lambda(\sigma,a)_i\|X_i-a\|^2 \\
&\qquad=
\sum_{\{i:X_i\text{ is a vertex of }\sigma\}}
\lambda(\sigma,a)_i r_{\mathrm{meb}}(\sigma)^2 \\
&\qquad=r_{\mathrm{meb}}(\sigma)^2.
\end{align}
Thus,
\begin{equation}
\sup_{x\in\sigma}
\sum_{i=1}^{n+1}
\delta_{\mathrm{del}}(\mathbf X^{n+1},x)_i\|X_i-x\|^2
=r_{\mathrm{meb}}(\sigma)^2.
\end{equation}

The simplices in $T(\mathbf X^{n+1})$ cover $\operatorname{conv}(\mathbf X^{n+1})=\operatorname{conv}(\mathbf X^n)$. Consequently,
\begin{align}
&\sup_{x\in\operatorname{conv}(\mathbf X^n)}
\sum_{i=1}^{n+1}
\delta_{\mathrm{del}}(\mathbf X^{n+1},x)_i\|X_i-x\|^2 \\
&\qquad=
\max_{\sigma\in T(\mathbf X^{n+1})}
\sup_{x\in\sigma}
\sum_{i=1}^{n+1}
\delta_{\mathrm{del}}(\mathbf X^{n+1},x)_i\|X_i-x\|^2 \\
&\qquad=
\max_{\sigma\in T(\mathbf X^{n+1})}
r_{\mathrm{meb}}(\sigma)^2.
\end{align}
Multiplying by $c/2$ establishes the first claim.

It remains to show that the search for an optimal collection location can be restricted to $B(\sigma_n^\star)$. First, adding a donor cannot increase the pointwise minimax matching loss: any convex affine-exact weights for the original design remain feasible for the augmented design after assigning weight zero to the new donor. Since Delaunay weights minimize that loss by Theorem \ref{thm:delaunay_budget},
\begin{equation}
\max_{\sigma\in T(\mathbf X^{n+1})}r_{\mathrm{meb}}(\sigma)^2
\leq
\max_{\sigma\in T(\mathbf X^n)}r_{\mathrm{meb}}(\sigma)^2.
\end{equation}

Now let $z\notin B(\sigma_n^\star)$. The circumball of $\sigma_n^\star$ remains empty after adding $z$. Therefore, the same empty-circumball supporting-hyperplane argument used in the proof of Theorem \ref{thm:delaunay_budget} shows that the augmented pointwise loss at the center of the minimum enclosing ball of $\sigma_n^\star$ remains $r_{\mathrm{meb}}(\sigma_n^\star)^2$. Consequently,
\begin{align}
\max_{\sigma\in T(\mathbf X^{n+1})}r_{\mathrm{meb}}(\sigma)^2
&\geq r_{\mathrm{meb}}(\sigma_n^\star)^2 \\
&=\max_{\sigma\in T(\mathbf X^n)}r_{\mathrm{meb}}(\sigma)^2.
\end{align}
Combining the last two displays shows that every $z$ outside $B(\sigma_n^\star)$ leaves the collection loss unchanged. Hence any strict improvement must lie in $\operatorname{conv}(\mathbf X^n)\cap B(\sigma_n^\star)$. If no strict improvement is possible, the restricted region still contains feasible points and attains the unchanged minimum. Substituting the first equality into Definition \ref{def:collection_plan} therefore establishes the stated restricted minimization problem.
\end{proof}

\end{document}